\documentclass[10pt]{article}

\usepackage{graphicx}
\usepackage{amsmath,amssymb}
\usepackage{amsmath}
\usepackage{amssymb}
\usepackage{amsfonts}
\usepackage{tabularx}
\usepackage{amsthm}
\usepackage{cite}
\usepackage{multirow}
\usepackage{extarrows}

\DeclareMathOperator{\Tr}{Tr}
\DeclareMathOperator{\Ex}{Ex}
\DeclareMathOperator{\Var}{Var}

\renewcommand{\Re}{{\rm Re}}
\renewcommand{\Im}{{\rm Im}}

\newcommand{\C}{{\mathbb C}}

\newcommand{\KK}{{\cal K}}

\newcommand{\HH}{{\cal H}}

\newcommand{\TT}{{\cal T}}

\newcommand{\LLL}{{\cal L}}
\newcommand{\EE}{{\cal E}}

\newcommand{\ID}{{\bf 1}}

\newcommand{\df}{\delta }

\newcommand{\ket}[1]{|{#1}\rangle}
\newcommand{\kb}[1]{|{#1}\rangle\!\langle{#1}|}
\newcommand{\bk}[1]{\langle{#1}|{#1}\rangle}
\newcommand{\bra}[1]{\langle{#1}|}
\newcommand{\bkt}[2]{\langle{#1}|{#2}\rangle}
\newcommand{\kbt}[2]{|{#1}\rangle\!\langle{#2}|}
\newcommand{\Ref}[1]{(\ref{#1})}
\newcommand{\wv}[2]{\langle{#1}\rangle_{#2}}

\newtheorem{thm}{Theorem}[section]
\newtheorem{col}[thm]{Corollary}

\begin{document}
\title{Theory of ``Weak Value" and Quantum Mechanical Measurements}
\author{Yutaka Shikano~\thanks{email: yshikano@ims.ac.jp} \\ Department of Physics, Tokyo Institute of Technology, Tokyo, Japan~\thanks{My current 
affiliation is Institute for Molecular Science located at Okazaki, Aichi, Japan.} \\
Center of Quantum Studies, Schmid College of Science and Technology, \\ Chapman University, CA, USA}
\date{\today}
\maketitle
\tableofcontents
\section{Introduction}
	Quantum mechanics provides us many perspectives and insights on Nature and our daily life. However, 
	its mathematical axiom initiated by von Neumann~\cite{neumann55} is not satisfied to describe nature phenomena. 
	For example, it is impossible not to explain a non self-adjoint operator, i.e., the momentum operator on a half line (See, e.g., Ref.~\cite{SH2}.), 
	as the physical observable. On considering foundations of quantum mechanics, the simple and specific expression is needed. One of the candidates is 
	the {\it weak value} initiated by Aharonov and his colleagues~\cite{AAV}. It is remarked that the idea of their seminal work is written in Ref.~\cite{AACV}. 
	Furthermore, this quantity has a potentiality 
	to explain the counter-factual phenomena, in which there is the contradiction under the 
	classical logic, e.g., the Hardy paradox~\cite{HARDY}. If so, it may be possible to quantitatively explain 
	quantum mechanics in the particle picture. In this review based on the author thesis~\cite{shikano_phd}, we consider the theory of the weak value and construct 
	a measurement model to extract the weak value. See the other reviews in Refs.~\cite{AR,AT,AV08,ATS}.
	
	Let the weak value for an observable $A$ be defined as 
	\begin{equation}
		\,_{f}\wv{A}{i}^{w} := \frac{\bra{f} A \ket{i}}{\bkt{f}{i}},
	\end{equation}
	where $\ket{i}$ and $\ket{f}$ are called a pre- and post-selected state, respectively. As the naming of 
	the ``weak value", this quantity is experimentally accessible by the weak measurement as explained below. 
	As seen in Fig.~\ref{wvfigfig}, the weak value 
	can be measured as the shift of a meter of the probe after the weak interaction 
	between the target and the probe with the specific post-selection of the target. 
	Due to the weak interaction, the quantum state of the target is only slightly changed but 
	the information of the desired observable $A$ is encoded in the probe by the post-selection.
	While the previous studies of the weak value since the seminal paper~\cite{AAV}, 
	which will be reviewed in Sec.~\ref{wvrev_sec}, 
	are based on the measurement scheme, there are few works that the weak value is focused on and 
	is independent of the measurement scheme. Furthermore, in these 20 years, we have not yet understood 
	the mathematical properties of the weak value. In this chapter, we review the historical backgrounds 
	of the weak value and the weak measurement and recent development on the measurement model 
	to extract the weak value.
\begin{figure}[ht]
	\begin{center}
	\includegraphics[width=10cm]{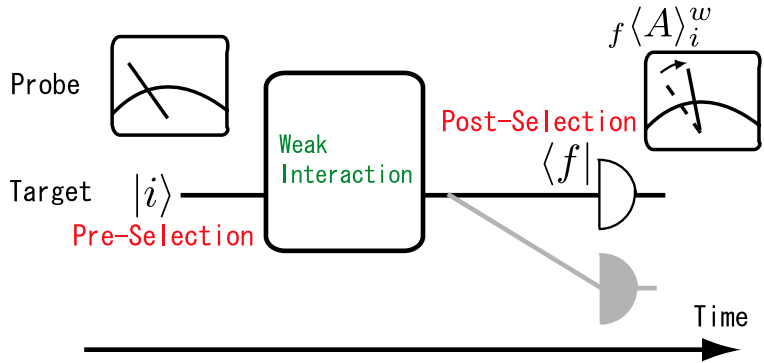}
	\caption{Schematic figure of the weak measurement.}
	\label{wvfigfig}
	\end{center}
\end{figure}

\section{Review of Quantum Operation} \label{qo_sec}
The time evolution for the quantum state and the operation for the measurement 
are called a quantum operation. In this section, we review a general description of the quantum operation.
Therefore, the quantum operation can describe the time evolution for 
the quantum state, the control of the quantum state, the quantum measurement, and the noisy 
quantum system in the same formulation.
\subsection{Historical Remarks} \label{dqm_sec}
Within the mathematical postulates of quantum mechanics~\cite{neumann55}, 
the state change is subject to the Schr\"{o}dinger equation. However, 
the state change on the measurement is not subject to this but is subject to 
another axiom, conventionally, von Neumann-L\"{u}ders projection postulate~\cite{Luders}. 
See more details on quantum measurement theory in the books~\cite{braginsky92, busch91, WM10}.

Let us consider a state change from the initial state $\ket{\psi}$ 
on the {\it projective measurement}~\footnote{This measurement is often called the 
{\it von Neumann measurement} or the {\it strong measurement}.} for the operator 
$A = \sum_j a_j \kb{a_j}$. From the Born rule, the probability to obtain the 
measurement outcome, that is, the eigenvalue of the observable $A$, is given by
\begin{equation}
	\Pr [A = a_m] = | \bkt{a_m}{\psi} |^2 = \Tr \left[ \kb{\psi} \cdot \kb{a_m} \right] = \Tr \rho P_{a_m},
\end{equation}
where $\rho := \kb{\psi}$ and $P_{a_m} = \kb{a_m}$. 
After the measurement with the measurement outcome $a_m$, the quantum state change is given by 
\begin{equation}
	\ket{\psi} \to \ket{a_m},
\end{equation}
which is often called the ``collapse of wavefunction" or ``state reduction". This implies 
that it is necessary to consider the {\bf non-unitary process} even in the isolated system. 
To understand the measuring process as quantum dynamics, we need consider the general theory 
of quantum operations. 
\subsection{Operator-Sum Representation}
Let us recapitulate the general theory of quantum operations of a 
finite dimensional quantum system~\cite{NC}. All physically realizable 
quantum operations can be generally described by a completely positive (CP) map~\cite{ozawa84,ozawa89}, 
since the isolated system of a target system and an auxiliary system always undergoes 
the unitary evolution according to the axiom of quantum mechanics~\cite{neumann55}. 
Physically speaking, the operation of the target system should be described as a 
positive map, that is, the map from the positive operator to the positive operator, since 
the density operator is positive. Furthermore, if any auxiliary system is coupled to the 
target one, the quantum dynamics in the compound system should be also described as the 
positive map since the compound system should be subject to quantum mechanics. Given the 
positive map, the positive map is called a CP map if and only if the positive map 
is also in the compound system coupled to any auxiliary system.
One of the important aspects of the CP map is that all 
physically realizable quantum operations can be described only by operators 
defined in the target system. Furthermore, the auxiliary system can be environmental system, 
the probe system, and the controlled system. Regardless to the role of the 
auxiliary system, the CP map gives the same description for the target system. 
On the other hand, both quantum measurement and decoherence give the same 
role for the target system. 

Let $\EE$ be a positive map from $\LLL(\HH_s)$, a set of linear operations on 
the Hilbert space $\HH_s$, to $\LLL(\HH_s)$. 
If $\EE$ is completely positive, its trivial extension $\KK$ from $\LLL(\HH_s)$ to $\LLL(\HH_s \otimes \HH_e)$ is also positive such that 
\begin{equation}
	\KK (\ket{\alpha}) := (\EE \otimes \ID) (\kb{\alpha}) > 0,
	\label{sig}
\end{equation}
for an arbitrary state $\ket{\alpha} \in \HH_s \otimes \HH_p$, where $\ID$ is the identity operator. 
We assume without loss of generality ${\rm dim}\HH_s = {\rm dim}\HH_e < \infty$. Throughout this chapter, we concentrate 
on the case that the target state is pure though the generalization to mixed states 
is straightforward. From the complete positivity, we obtain the following theorem for 
quantum state changes.
\begin{thm}
Let $\EE$ be a CP map from $\HH_s$ to $\HH_s$.
For any quantum state $\ket{\psi}_s \in \HH_s$, there exist a map $\sigma$ and a pure state 
$\ket{\alpha} \in \HH_s \otimes \HH_e$ such that 
\begin{equation}
	\EE (\ket{\psi}_s \bra{\psi})= \,_e\bra{\tilde{\psi}} \KK (\ket{\alpha}) \ket{\tilde{\psi}}_e,
	\label{schmidt}
\end{equation}
where
\begin{equation}
	\ket{\psi}_s = \sum_{k} \psi_{k} \ket{k}_{s}, \ \ \ 
	\ket{\tilde{\psi}}_e = \sum_{k} \psi^{\ast}_{k} \ket{k}_{e},
\end{equation}
which represents the state change for the density operator.
\label{1th}
\end{thm}
\begin{proof}
We can write in the Schmidt form as
\begin{equation}
	\ket{\alpha} = \sum_m \ket{m}_s \ket{m}_e.
\end{equation}
We rewrite the right hand sides of Eq. \Ref{schmidt} as 
\begin{align}
	\KK (\ket{\alpha}) & = (\EE \otimes \ID) \left( \sum_{m,n} \ket{m}_{s} \ket{m}_{e} 
	\,_{s}\bra{n} _{e}\bra{n} \right) \notag \\
	& = \sum_{m,n} \ket{m}_{e} \bra{n} \EE (\ket{m}_s \bra{n}),
\end{align}
to obtain
\begin{equation}
	\,_{e}\bra{m} \KK (\ket{\alpha}) \ket{n}_{e} = \EE (\ket{m}_{s} \bra{n}).
\end{equation}
By linearity, the desired equation (\ref{schmidt}) can be derived.
\end{proof}

From the complete positivity, $\KK (\ket{\alpha}) > 0$ for all 
$\ket{\alpha} \in \HH_s \otimes \HH_e$, we can express $\sigma (\ket{\alpha})$ as 
\begin{equation} \label{complete}
	\KK (\ket{\alpha}) = \sum_{m} s_{m} \kb{\hat{s}_{m}} = \sum_{m} \kb{s_{m}},
\end{equation}
where $s_m$'s are positive and $\{ \ket{\hat{s}_m} \}$ is a complete orthonormal 
set with $\ket{s_m}:=\sqrt{s_m} \ket{\hat{s}_m}$. 
We define the {\it Kraus operator} $E_m$~\cite{KRAUS} as 
\begin{equation}
	E_{m} \ket{\psi}_s := \,_{e}\bkt{\tilde{\psi}}{s_m}.
\end{equation}
Then, the quantum state change becomes the operator-sum representation, 
\begin{equation}
	\sum_{m} E_m \ket{\psi}_s \bra{\psi} E^{\dagger}_{m} = \sum_{m} 
	\,_{e}\bkt{\tilde{\psi}}{s_m} \bkt{s_m}{\tilde{\psi}}_{e}
	= \,_{e}\bra{\tilde{\psi}} \KK (\ket{\alpha}) \ket{\tilde{\psi}}_{e} \notag \\
	= \EE (\ket{\psi}_{s} \bra{\psi}).
\end{equation}
It is emphasized that the quantum state change is described solely in 
terms of the quantities of the target system.
\subsection{Indirect Quantum Measurement}~\label{iqm_sec}
In the following, the operator-sum representation of the quantum state change is related to the indirect 
measurement model. Consider the observable $A_s$ and $B_p$ for the target and probe systems given by 
\begin{equation}
	A_s = \sum_j a_j \ket{a_j}_s \bra{a_j}, \ \ \ 
	B_p = \sum_j b_j \ket{b_j}_p \bra{b_j},
\end{equation}
respectively. We assume that the interaction Hamiltonian is given by 
\begin{equation}
	H_{int} (t) = g ( A_s \otimes B_p) \ \df (t - t_0), 
\end{equation}
where $t_0$ is measurement time. Here, without loss of generality, the interaction 
is impulsive and the coupling constant $g$ is scalar. The quantum dynamics for 
the compound system is given by 
\begin{equation}
	\kb{s_{m}} = U (\ket{\psi}_{s} \bra{\psi} \otimes \ket{\phi}_{p} \bra{\phi}) U^{\dagger},
\end{equation}
where $\ket{\psi}_{s}$ and $\ket{\phi}_{p}$ are the initial quantum state on the target and 
probe systems, respectively. For the probe system, we perform the projective measurement for 
the observable $B_p$. The probability to obtain the measurement outcome $b_m$ is given by 
\begin{align}
	\Pr [B_p = b_m] & = \Tr_s \bra{b_m} U (\ket{\psi}_{s} \bra{\psi} \otimes \ket{\phi}_{p} 
	\bra{\phi}) U^{\dagger} \ket{b_m}, \notag \\
	& = \Tr_s E_m \ket{\psi}_{s} \bra{\psi} E^{\dagger}_m = \Tr_s \ket{\psi}_{s} \bra{\psi} M_m,
\end{align}
where the Kraus operator $E_m$ is defined as 
\begin{equation} \label{krausreview}
	E_m := \,_p\bra{b_m} U \ket{\phi}_p,
\end{equation}
and $M_m : = E^{\dagger}_m E_m$ is called a {\it positive operator valued measure} 
(POVM)~\cite{Davies70}. The POVM has the same role of the spectrum of 
the operator $A_s$ in the case of the projective measurement. To derive the projective measurement 
from the indirect measurement, we set the spectrum of the operator $A_s$ as the POVM, 
that is, $M_m = \ket{a_m}_s \bra{a_m}$. Since the sum of the probability distribution over the measurement outcome equals to one, 
we obtain 
\begin{align}
	\sum_m \Pr [B_p = b_m] = 1 
	& \Longleftrightarrow \sum_m \Tr \ket{\psi}_{s} \bra{\psi} M_m = \Tr \ket{\psi}_{s} \bra{\psi} \sum_m M_m = 1 \notag \\
	& \to \sum_m M_m = \ID . \label{preconsist}
\end{align}
Here, the last line uses the property of the density operator, $\Tr \ket{\psi}_{s} \bra{\psi} = 1$ for any $\ket{\psi}$.

\section{Review of Weak Value} \label{wvrev_sec}
In Secs.~\ref{dqm_sec} and~\ref{iqm_sec}, the direct and indirect quantum measurement schemes, 
we only get the probability distribution. However, the probability distribution is not the only 
thing that is experimentally accessible in quantum mechanics. 
In quantum mechanics, the phase is also an essential ingredient and in particular the geometric phase is a notable 
example of an experimentally accessible quantity~\cite{SW}. The general experimentally accessible quantity which 
contains complete information of the probability and the phase seems to be the {\it weak value} advocated by 
Aharonov and his collaborators ~\cite{AAV,AR}. They proposed a model of weakly coupled system and probe, see 
Sec.~\ref{WMR}, to obtain information to a physical quantity as a ``weak value" only slightly disturbing the state. 
Here, we briefly review the formal aspects of the weak value.

For an observable $A$, the {\it weak value} $\wv{A}{w}$ is defined as 
\begin{equation}
	\wv{A}{w} := \frac{{\bra{f}}U(t_f,t)AU(t,t_i)\ket{i}}{\bra{f}U(t_f,t_i)\ket{i}}\in \C,
	\label{wvdefi}
\end{equation} 
where $\ket{i}$ and $\bra{f}$ are normalized pre-selected ket and 
post-selected bra state vectors, respectively~\cite{AAV}. 
Here, $U(t_2,t_1)$ is an evolution operator from the time $t_1$ to $t_2$. 
The weak value $\wv{A}{w}$ actually depends on the pre- and post-selected states 
$\ket{i}$ and $\bra{f}$ but we omit them for notational simplicity in the case that 
we fix them. Otherwise, we write them explicitly as $_{f}\wv{A}{i}^{w}$ instead for $\wv{A}{w}$. 
The denominator is assumed to be non-vanishing. This quantity is, in general, in the 
complex number $\C$. Historically, the terminology ``weak value" 
comes from the {\it weak measurement}, where the coupling between the target system 
and the probe is weak, explained in the following section. Apart from their original 
concept of the weak value and the weak measurement, we emphasize that the concept of 
the weak value is independent of the weak measurement~\footnote{This concept is shared in 
Refs.~\cite{Johansen04, Johansen07, Mitchison08, Aberg09, Hu90, Parks10, Dressel10, Dressel11_2}.}. To take the weak 
value as {\it a priori} given quantity in quantum mechanics, we will construct 
the observable-independent probability space. In the conventional quantum measurement theory, 
the probability space, more precisely speaking, the probability measure, depends on the 
observable~\cite[Sec. 4.1]{shikano_master}~\footnote{Due to this, the probability in quantum 
mechanics cannot be applied to the standard probability theory. As another approach 
to resolve this, there is the quantum probability theory~\cite{Summers}.}.

Let us calculate the expectation value in quantum mechanics for the quantum state $\ket{\psi}$ as 
\begin{align}
	\Ex [A] = \bra{\psi}A\ket{\psi} 
	& = \int d \phi \, \bkt{\psi}{\phi} \bra{\phi} A \ket{\psi}  
	= \int d \phi \, \bkt{\psi}{\phi} \cdot \bkt{\phi}{\psi} \frac{\bra{\phi} A \ket{\psi}}{\bkt{\phi}{\psi}},
	\notag \\
	& = \int d \phi \, |\bkt{\psi}{\phi}|^2 \,_{\phi}\wv{A}{\psi}^{w}, \label{expwv}
\end{align}
where $h_{A} [\ket{\phi}] =  \,_{\phi}\wv{A}{\psi}^{w}$ is complex random variable and 
$dP := | \bkt{\phi}{\psi} |^2 d \phi$ is the probability measure and is independent of the 
observable $A$. Therefore, the event space $\Omega = \{ \ket{\phi} \}$ is taken as the set of the 
post-selected state. This formula means that the extended probability theory corresponds to the 
Born rule. From the conventional definition of the variance in 
quantum mechanics, we obtain the variance as 
\begin{align}
	\Var [A] & = \int |h_A [\ket{\phi}]|^2 dP - \left( \int h_A [\ket{\phi}] dP \right)^2 \notag \\
	& = \int \left| \frac{\bra{\phi} A \ket{\psi}}{\bkt{\phi}{\psi}} \right|^2 |\bkt{\phi}{\psi}|^2 d \phi 
	- \left( \int \frac{\bra{\phi} A \ket{\psi}}{\bkt{\phi}{\psi}} |\bkt{\phi}{\psi}|^2 d \phi 
	\right)^2 \notag \\
	& = \int \left| \bra{\phi} A \ket{\psi} \right|^2 d \phi - \left( \int \bkt{\psi}{\phi} 
	\bra{\phi} A \ket{\psi} d \phi \right)^2 \notag \\
	& = \int \bra{\psi} A \kb{\phi} A \ket{\psi} d \phi - (\bra{\psi} A \ket{\psi})^2 \notag \\
	& = \bra{\psi} A^2 \ket{\psi} - (\bra{\psi} A \ket{\psi})^2. \label{varwv}
\end{align}
This means that the observable-independent probability space can be characterized by the weak value~\cite{SH}. 
From another viewpoint of the weak value, the statistical average of the weak value 
coincides with the expectation value in quantum mechanics~\cite{AB}. This can be interpreted as the 
probability while this allows the ``negative probability"~\footnote{The concept of negative probability 
is not new, e.g., see Refs.~\cite{Dirac,FEYNMAN, holger1, holger5, HOFMANN}. The weak value defined by Eq. (\ref{wvdefi}) is 
normally called the transition amplitude from the state $\ket{\psi}$ to $\bra{\phi}$ via the 
intermediate state $\ket{a}$ for $A=\kb{a}$, the absolute value squared of which is the probability 
for the process. But the three references quoted above seem to suggest that they might be 
interpreted as probabilities in the case that the process is counter-factual, i.e., the 
case that the intermediate state $\ket{a} $ is not projectively measured. The description 
of intermediate state $\ket{a}$ in the present work is counter-factual or virtual 
in the sense that the intermediate state would not be observed by projective measurements. Feynman's 
example is the counter-factual ``probability" for an electron to have its spin up in the 
$x$-direction and also spin down in the $z$-direction~\cite{FEYNMAN}.}. On this idea, 
the uncertainty relationship was analyzed on the Robertson 
inequality~\cite{Garretson04, Sokolovski07} and on 
the Ozawa inequality~\cite{LW10}, which the uncertainty relationships are 
reviewed in Ref.~\cite[Appendix A]{shikano_master}.
Also, the joint probability for the compound system was analyzed in Refs.~\cite{Berry11, Botero08}.
Furthermore, if the operator $A$ is a projection operator $A=\kb{a}$, the above identity 
becomes an analog of the Bayesian formula,
\begin{equation}
|\bkt{a}{\psi}|^2 = \int \,_{\phi}\wv{\kb{a}}{\psi}^{w} |\bkt{\phi}{\psi}|^2 d \phi.
\label{bay}
\end{equation} 
The left hand side is the probability to obtain the state $\ket{a}$ given the 
initial state $\ket{\psi}$. From this, one may get some intuition by interpreting 
the weak value $\,_{\phi}\wv{\kb{a}}{\psi}^{w}$ as the complex conditional probability 
of obtaining the result $\ket{a}$ under an initial condition $\ket{i}$ and a 
final condition $\ket{f}$ in the process $\ket{i} \rightarrow \ket{a} \rightarrow 
\ket{f}$~\cite{Steinberg1, Steinberg2}~\footnote{The interpretation of the weak value 
as a complex probability is suggested in the literature \cite{MJP}.}.
Of course, we should not take the strange weak values too literally but the remarkable 
consistency of the framework of the weak values due to Eq. (\ref{bay}) and a consequence 
of the completeness relation,
\begin{equation}
\sum_{a}\wv{\ket{a}\bra{a}}{w}=1,
\label{comp}
\end{equation}
may give a useful concept to further push theoretical consideration by intuition.

This interpretation of the weak values gives many possible examples of strange phenomena 
like a negative kinetic energy~\cite{APRV}, a spin $100 \hbar$ for an 
electron~\cite{AAV, Duck89, Golub89, Ashhab09} and a superluminal 
propagation of light~\cite{RY, Sokolovski05} and neutrino~\cite{Tanimura, Berry11_2} motivated by the 
OPERA experiment~\cite{OPERA}. 
The framework of weak values has been theoretically applied to foundations of quantum physics, e.g., 
the derivation of the Born rule from the alternative assumption for {\it a priori} measured value~\cite{Hosoya11}, 
the relationship to the uncertainty relationship~\cite{holger6}, 
the quantum stochastic process~\cite{WANG}, the tunneling traverse 
time~\cite{Steinberg1, Steinberg2, Ranfagni93}, arrival time and time 
operator~\cite{Ruseckas02, Ahnert04, Ruseckas04, Busch90}, the decay law~\cite{Davies09, Urbanowski09}, 
the non-locality~\cite{Tollaksen10, Tollaksen07, Brodutch08}, 
especially, quantum non-locality, 
which is characterized by the modular variable, consistent history~\cite{VAIDMAN, Kastner04}, 
Bohmian quantum mechanics~\cite{Leavens05}, semi-classical weak values on the tunneling~\cite{TANAKA}, 
the quantum trajectory~\cite{WISEMAN}, and classical stochastic theory~\cite{Tomita}. Also, in quantum information science, the weak value was 
analyzed on quantum computation~\cite{Oreshkov05, Brun08}, quantum communications~\cite{BACGS, Botero00}, 
quantum estimation, e.g., state tomography~\cite{holger2, Hofmann11, Hofmann11_2, Shpitalnik08, Massar11} 
and the parameter estimation~\cite{holger3, holger4, ST}, 
the entanglement concentration~\cite{Menzies07}, the quasi-probability distribution~\cite{Bednorz, Sagawa, Tsang, Gosson} and 
the cloning of the unknown quantum state with hint~\cite{Sjoqvist06}. Furthermore, this was applied to the cosmological situations in 
quantum-mechanical region, e.g., the causality~\cite{Anandan02}, the inflation theory~\cite{Campo04}, backaction of 
the Hawking radiation from the black hole~\cite{Englert95, Englert10, Brout95}, and the new interpretation 
of the universe~\cite{AG05, Gruss00, Ellis10}. However, the most important fact is that the weak 
value is experimentally accessible so that the intuitive argument based on the weak values 
can be either verified or falsified by experiments. There are many experimental proposals to obtain the 
weak value in the optical~\cite{Li11, Knight90, Simon11, Agarwal07, Cho10, Menzies09, Wu11} and the 
solid-state~\cite{Romito07, Romito10, Korotkov06, Miller09, Jordan07, Williams08, Jordan10, Zilberberg11} systems. 
Recently, the unified viewpoint was found in the weak measurement~\cite{Kofman}.

On the realized experiments on the weak value, we can classify the three concepts: (i) 
testing the quantum theory, (ii) the amplification of the tiny effect in quantum mechanics, 
and (iii) the quantum phase. 
\begin{itemize}
	\item[(i)] Testing the quantum theory. The weak value can solve many quantum paradoxes 
	seen in the book~\cite{AR}. The Hardy paradox~\cite{HARDY}, which there occurs in 
	two Mach-Zehnder interferometers of the electron and the position, was resolved by 
	the weak value~\cite{hardy_aharonov} and was analyzed deeper~\cite{HS}. This paradoxical situation 
	was experimentally demonstrated in the optical setup~\cite{Lundeen09, YYKI}.
	By the interference by the polarization~\cite{PCS} 
	and shifting the optical axis~\cite{RSH}, the spin beyond the eigenvalue is 
	verified. By the latter technique, the three-box paradox~\cite{AV90, VAIDMAN} 
	was realized~\cite{RLS}. Thereafter, the theoretical progresses are 
	the contextuality on quantum mechanics~\cite{Tollaksen07_3}, the generalized 
	N-box paradox~\cite{Leavens06}, and the relationship 
	to the Kirkpatrick game~\cite{Ravon07}. The weak value is used to show the violation 
	of the Leggett-Garg inequality~\cite{Williams08, Marcovitch11}. This experimental 
	realizations were demonstrated 
	in the system of the superconducting qubit~\cite{Laloy}, the optical systems~\cite{Pryde2, Dressel11}.
	Furthermore, since the weak value for the position observable $\kb{x}$ with the pre-selected 
	state $\ket{\psi}$ and the post-selection $\ket{p}$ is given by  
	\begin{equation}
		\wv{\kb{x}}{w} = \frac{\bkt{p}{x}\bkt{x}{\psi}}{\bkt{p}{\psi}} = \frac{e^{ixp} 
		\psi (x)}{\phi (p)},
	\end{equation}
	we obtain the wavefunction $\psi (x) := \bkt{x}{\psi}$ as the weak value with the multiplication 
	factor $1/\phi (0)$ with $\phi (p) := \bkt{p}{\psi}$ in the case of $p = 0$. Using the photon 
	transverse wavefunction, there are experimentally demonstrated by replacing the weak measurement for the 
	position as the polarization measurement~\cite{Lundeen11}. This paper was theoretically criticized 
	to compare the standard quantum state tomography for the phase space in Ref.~\cite{Happasalo} and was 
	generalized to a conventionally unobservable~\cite{Lundeen11_2}.
	As other examples, there are the detection of the superluminal signal~\cite{Brunner04}, the quantum 
	non-locality~\cite{Spence10}, and the Bohmian trajectory~\cite{Shalm10, Kocsis11} on the base of the 
	theoretical analysis~\cite{Wiseman07}.
	\item[(ii)] Amplification of the tiny effect in quantum mechanics. Since the weak value has the denominator, the weak value is very large 
	when the pre- and post-selected states are almost orthogonal\footnote{Unfortunately, the signal to noise ratio is not drastically changed under the 
	assumption that the probe wavefunction is Gaussian on a one-dimensional parameter space.}. This is practical advantage to use 
	the weak value. While the spin Hall effect of light~\cite{Onoda04} is too tiny effect to 
	observe its shift in the conventional scheme, by almost orthogonal polarizations 
	for the input and output, this effect was experimentally verified~\cite{HK} to be theoretically 
	analyzed from the viewpoint of the spin moments~\cite{Krowne09}. Also, some interferometers were applied. 
	The beam deflection on the Sagnac interferometer~\cite{Dixon09} was shown to be supported by 
	the classical and quantum theoretical analyses~\cite{Howell10}~\footnote{Unfortunately, 
	the experimental data are mismatched to the theoretical prediction. While the authors claimed that this differences results from the stray of light, 
	the full-order calculation even is not mismatched~\cite{Koike_full}. However, this difference remains the open problem.}. Thereafter, 
	optimizing the signal-to-noise ratio~\cite{Turner11, Starling09}, the phase 
	amplification~\cite{Starling10, Starling10_2}, and the precise frequency measurement~\cite{Starling10_3} 
	were demonstrated. As another example, there is shaping the laser pulse beyond the 
	diffraction limit~\cite{Ranfagni04}. According to Steinberg~\cite{Steinberg_private}, in his group, 
	the amplification on the single-photon nonlinearity has been progressed to be based on the theoretical 
	proposal~\cite{Feizpour11}. While the charge sensing amplification was proposed 
	in the solid-state system~\cite{Zilberberg11}, 
	there is no experimental demonstration on the amplification 
	for the solid-state system. Furthermore, the upper bound of the amplification has not yet solved. 
	Practically, this open problem is so important to understand the relationship to the 
	weak measurement regime. 
	\item[(iii)] Quantum phase. The argument of the weak value for the projection 
	operator is the geometric phase as 
	\begin{align}
		\gamma & := \arg \bkt{\psi_1}{\psi_2} \bkt{\psi_2}{\psi_3} \bkt{\psi_3}{\psi_1} \notag \\
			& = \arg \frac{\bkt{\psi_1}{\psi_2} \bkt{\psi_2}{\psi_3} \bkt{\psi_3}{\psi_1}}{|\bkt{\psi_3}{\psi_1}|^2} 
			= \arg \frac{\bkt{\psi_1}{\psi_2} \bkt{\psi_2}{\psi_3}}{\bkt{\psi_1}{\psi_3}} \notag \\
			& = \arg \,_{\psi_1}\wv{\kb{\psi_2}}{\psi_3}^{w}. \label{geowv}
	\end{align}
	where the quantum states, $\ket{\psi_1}, \ket{\psi_2}$, and $\ket{\psi_3}$, are the pure states~\cite{Sjoqvist}. 
	Here, the quantum states, $\ket{\psi_1}$ and $\ket{\psi_3}$, are the post- and pre-selected states, respectively.
	Therefore, we can evaluate the weak value 
	from the phase shift~\cite{Tamate09}. Of course, vice versa~\cite{Brunner10}. Tamate {\it et al.} 
	proposal was demonstrated on the relationship to quantum eraser~\cite{Kobayashi11} and by the 
	a three-pinhole interferometer~\cite{Kobayashi10}. The phase shift from the zero mode to $\pi$ mode 
	was observed by using the interferometer with a Cs vapor~\cite{Camacho09} and the phase shift in 
	the which-way path experiment was demonstrated~\cite{Mir07}. 
	Furthermore, by the photonic crystal, phase singularity was demonstrated~\cite{Solli04}.
	\item[(iv)] Miscellaneous. The backaction of the weak measurement is experimentally realized in the optical system~\cite{Iinuma}. 
	Also, the parameter estimation using the weak value is demonstrated~\cite{holger4}.
\end{itemize}
\section{Historical Background -- Two-State Vector Formalism} \label{wvhis_sec}
In this section, we review the original concept of the two-state vector formalism. 
This theory is seen in the reviewed papers~\cite{AV08, AT}.
\subsection{Time Symmetric Quantum Measurement}
While the fundamental equations of the microscopic physics are time symmetric, for example, the 
Newton equation, the Maxwell equation, and the Schr\"{o}dinger equation~\footnote{It is, of course, noted that 
thermodynamics does not have the time symmetric properties from the second law of thermodynamics.}, 
the quantum measurement is not time symmetric. This is because the quantum state after 
quantum measurement depends on the measurement outcome seen in Sec.~\ref{qo_sec}. 
The fundamental equations of the microscopic 
physics can be solved to give the initial boundary condition. To construct the time 
symmetric quantum measurement, 
the two boundary conditions, which is called pre- and post-selected states, are needed. The concept of the 
pre- and post-selected states is called the two-state vector formalism~\cite{ABL}. In the following, 
we review the original motivation to construct the time symmetric quantum measurement. 

Let us consider the projective measurement for the observable $A = \sum_i a_i \kb{a_i}$ 
with the initial boundary condition denoted as $\ket{i}$ at time $t_i$. 
To take quantum measurement at time $t_0$, the probability to obtain the measurement outcome $a_j$ is given by 
\begin{equation}
	\Pr [A = a_j] = \parallel \bra{a_j} U \ket{i} \parallel^2, 
\end{equation}
with the time evolution $U := U(t_0, t_i)$. After the projective measurement, the quantum state 
becomes $\ket{a_j}$. 
Thereafter, the quantum state at $t_f$ is given by $\ket{\varphi_j} := V \ket{a_j}$ with 
$V = U(t_f, t_0)$. the probability 
to obtain the measurement outcome $a_j$ can be rewritten as 
\begin{equation}
	\Pr [A = a_j] = \frac{\parallel \bra{\varphi_j} V \ket{a_j} \parallel^2 \parallel 
	\bra{a_j} U \ket{i} \parallel^2}{\sum_j \parallel \bra{\varphi_j} V \ket{a_j} 
	\parallel^2 \parallel \bra{a_j} U \ket{i} \parallel^2}. \label{projpro}
\end{equation}
It is noted that $\parallel \bra{\varphi_j} V \ket{a_j} \parallel^2 = 1$. 
Here, we consider the backward time evolution from the quantum state $\ket{\varphi_j}$ at time $t_f$. 
We always obtain the quantum state $\ket{a_j}$ after the projective measurement at time $t_0$. 
Therefore, the quantum state at time $t_i$ is given by 
\begin{equation}
	\ket{{\tilde i}} := U^{\dagger} \kb{a_j} V^{\dagger} \ket{\varphi_j} = U^{\dagger} \ket{a_j}.
\end{equation}
In general, $\ket{{\tilde i}}$ is different from $\ket{i}$. Therefore, projective measurement is time 
asymmetric. 

To construct the time-symmetric quantum measurement, we add the boundary condition at time $t_f$. 
Substituting the quantum state $\ket{\varphi_j}$ to the {\it specific} one denoted as $\ket{f}$, 
which is called the post-selected state, the probability to obtain the measurement outcome 
$a_j$, Eq.~(\ref{projpro}), becomes 
\begin{equation}
	\Pr [A = a_j] = \frac{\parallel \bra{f} V \ket{a_j} \parallel^2 \parallel 
	\bra{a_j} U \ket{i} \parallel^2}{\sum_j \parallel \bra{f} V \ket{a_j} 
	\parallel^2 \parallel \bra{a_j} U \ket{i} \parallel^2}. \label{ablrule}
\end{equation}
This is called the Aharonov-Bergmann-Lebowitz (ABL) formula~\cite{ABL}. 
From the analogous discussion to the above, this measurement is time symmetric.
Therefore, describing quantum mechanics by the pre- and post-selected states,  $\ket{i}$ and $\bra{f}$, is called the ``two-state vector formalism".
\subsection{Protective Measurement}
In this subsection, we will see the noninvasive quantum measurement for the specific 
quantum state on the target system. 
Consider a system of consisting of a target and a probe defined in the Hilbert space $\HH_s\otimes\HH_p$. 
The interaction between the target and the probe is given by 
\begin{equation}
	H_{int} (t) = g(t) (A \otimes \hat{P}),
\end{equation}
where  
\begin{equation}
	\int^{T}_{0} g(t) dt =: g_0.
\end{equation}
The total Hamiltonian is given by 
\begin{equation}
	H_{tot} (t) = H_s (t) + H_p (t) + H_{int} (t).
\end{equation}
Here, we suppose that $H_s (t)$ has discrete and non-degenerate eigenvalues denoted as $E_i (t)$. 
Its corresponding eigenstate is denoted as $\ket{E_i (t)}$ for any time $t$. Furthermore, we 
consider the discretized time from the time interval $[0, T]$; 
\begin{equation}
	t_n = \frac{n}{N} T \ (n = 0, 1, 2, \dots, N),
\end{equation}
where $N$ is a sufficiently large number. We assume that the initial target state is the energy eigenvalue $\ket{E_i (t)}$~\footnote{Due to 
this assumption, it is impossible to apply this to the arbitrary 
quantum state. Furthermore, while we seemingly need the projective measurement, that is, 
destructive measurement, for the target system to confirm whether the initial quantum state 
is in the eigenstates~\cite{Unruh_com, Rovelli_com}, they did not apply this to the arbitrary 
state. For example, if the system is cooled down, we can pickup the ground state of the target 
Hamiltonian $H_s (0)$.} the initial probe state 
is denoted as $\ket{\xi (0)}$. Under the adiabatic condition, the compound state 
for the target and probe systems at time $T$ is given by 
\begin{align}
	\ket{\Phi (T)} & := \kb{E_i (t_N)} e^{- i \frac{T}{N} H_{tot} (t_N)} \kb{E_i (t_{N-1})} e^{- i \frac{T}{N} H_{tot} (t_{N-1})} \cdots \notag \\
	& \ \ \ \ \ \times \kb{E_i (t_2)} e^{- i \frac{T}{N} H_{tot} (t_2)} \kb{E_i (t_1)} e^{- i \frac{T}{N} H_{tot} (t_1)} \ket{E_i (0)} \otimes 
	\ket{\xi (0)}.
\end{align}
Applying the Trotter-Suzuki theorem~\cite{TS1,TS2}, one has 
\begin{align}
	\ket{\Phi (T)} & := \kb{E_i (t_N)} e^{- i \frac{T}{N} H_{int} (t_N)} 
	\kbt{E_i (t_{N})}{E_i (t_{N-1})} e^{- i \frac{T}{N} H_{int} 
	(t_{N-1})} \cdots \notag \\
	& \ \ \ \ \ \times \kbt{E_i (t_3)}{E_i (t_2)} e^{- i \frac{T}{N} H_{int} (t_2)} \kbt{E_i (t_2)}{E_i (t_1)} 
	e^{- i \frac{T}{N} H_{int} (t_1)} \ket{E_i (1)} \otimes \ket{\xi (T)}.
\end{align}
By the Taylor expansion with the respect to $N$, the expectation value is 
\begin{align}
	\bra{E_i (t_n)} e^{-i \frac{T}{N} g(t_n) A \otimes \hat{P}} \ket{E_i (t_n)} & = 
	1 - i \frac{T}{N} g (t_n) \Ex [A(t_n)] \hat{P} - \frac{1}{2} \frac{T^2}{N^2} g^2 (t_n) 
	( \Ex [A(t_n)] )^2 \hat{P}^2 \notag \\
	& \ \ \ \ \ \ \ \ \ \ \ \ \ \ \ \ \ \ \ \ \ \ \ \ \ 
	- \frac{1}{2} \frac{T^2}{N^2} g^2 (t_n) \Var [A(t_n)] \hat{P}^2 + O \left( \frac{1}{N^3} \right) \notag \\
	& \sim e^{- i \frac{T}{N} g (t_n) \Ex [A(t_n)] \hat{P}} \left( 1 - \frac{1}{2} \frac{T^2}{N^2} 
	g^2 (t_n) \Var [A(t_n)] \hat{P}^2 \right).
\end{align}
In the limit of $N \to \infty$, by quadrature by parts, we obtain 
\begin{align}
	\ket{\Phi (T)} & \sim \ket{E_i (T)} \exp \left[ - i \left( \int^{T}_{0} g(t) \Ex [A(t)] 
	dt \right) \hat{P} \right] \notag \\ 
	& \ \ \ \ \ \ \ \ \ \ \ \ \ \ \ \ \ \ \ \ \times \left[ 1 
	- \frac{T}{N}  \left(\int^{T}_{0} g^2 (t) \Var [A(t)] dt \right) \hat{P}^2 \right] \ket{\xi (T)} 
	+ O \left( \frac{1}{N} \right) \notag \\
	& = \ket{E_i (T)} \exp \left[ - i \left( \int^{T}_{0} g(t) \Ex [A(t)] dt \right) \hat{P} \right] 
	\ket{\xi (T)}.
\end{align}
Therefore, the shift of the expectation value for the position operator on the probe system is given by 
\begin{equation}
	\Delta [Q] = \int^{T}_{0} g(t) \Ex [A(t)] dt.
\end{equation}
It is emphasized that the quantum state on the target system remains to be the energy eigenstate of $H_s$. 
Therefore, this is called the {\it protective measurement}~\cite{AV93, AAV93}. It is remarked 
that the generalized version of 
the protective measurement in Ref.~\cite{AV95} by the pre- and post-selected states and 
in Ref.~\cite{AMTPV} by the meta-stable state. 
\subsection{Weak Measurement} \label{WMR}
From the above discussions, is it possible to combine the above two concepts, i.e., the 
time-symmetric quantum measurement 
without destroying the quantum state~\cite{Vaidman96}? 
This answer is the {\it weak measurement}~\cite{AAV}. 
Consider a target system and a probe defined in the Hilbert space $\HH_s \otimes \HH_p$. 
The interaction of the target system and the probe is assumed to be weak and instantaneous,
\begin{equation} \label{wmint}
	H_{int} (t) = g (A \otimes \hat{P}) \df (t-t_0),
\end{equation}
where an observable $A$ is defined in $\HH_s$, while $\hat{P}$ is the momentum operator of the probe. 
The time evolution operator becomes $e^{-ig(A\otimes \hat{P})}$. 
Suppose the probe initial state is $\ket{\xi}$.
For the transition from the pre-selected state $\ket{i}$ to the post-selected state $\ket{f}$, 
the probe wave function becomes $\ket{\xi^{\prime}} = 
\bra{f}Ve^{-ig(A\otimes \hat{P})}U\ket{i} \ket{\xi}$, 
which is in the weak coupling case,
\begin{align}
\ket{\xi^{\prime}}&=\bra{f}V e^{-ig(A \otimes \hat{P})} U \ket{i} \ket{\xi} \notag \\ 
&= \bra{f}V [ \ID -ig(A \otimes \hat{P}) ] U \ket{i} \ket{\xi} + O (g^2) \notag \\
&= \bra{f}V U \ket{i} - ig \bra{f}V A U \ket{i} \otimes \hat{P} \ket{\xi} + O (g^2) \notag \\
&= \bra{f}V U \ket{i} \left( 1 - ig \wv{A}{w} \hat{P} \right) \ket{\xi} + O (g^2) \label{wmi}
\end{align}
where $\bra{f}VAU\ket{i}/\bra{f}VU\ket{i}=\wv{A}{w}$.
Here, the last equation uses the approximation that $g \wv{A}{w} \ll 1$~\footnote{It is remarked 
that Wu and Li showed the second-order correction of the weak measurement~\cite{Wu2}. A further analysis was 
shown in Refs.~\cite{Pan, Parks11}.}. 
We obtain the shifts of the expectation values for the position and momentum operators on the probe as 
the following theorem: 
\begin{thm}[Jozsa~\cite{JOZSA}] \label{Jozsaproof}
We obtain the shifts of the expectation values for the position and momentum operators on the probe 
after the weak measurement with the post-selection as 
\begin{align}
	\Delta [\hat{Q}] & = g \Re \wv{A}{w} + m g \Im \wv{A}{w}  
	\left. \frac{d \Var [\hat{Q}]}{dt} \right|_{t=t_0}, \label{qshift} \\
	\Delta [\hat{P}] & = 2 g \Im \wv{A}{w} \Var [\hat{P}], \label{pshift}
\end{align}
where 
\begin{align}
\Delta [\hat{Q}] & := \frac{\bra{\xi^\prime} \hat{Q} \ket{\xi^\prime}}{\bk{\xi^\prime}} 
- \bra{\xi} \hat{Q} \ket{\xi}, \\ 
\Delta [\hat{P}] &:= \frac{\bra{\xi^\prime} \hat{P} \ket{\xi^\prime}}{\bk{\xi^\prime}} - \bra{\xi} 
\hat{P} \ket{\xi}, \\ 
\Var [\hat{Q}] &:= \bra{\xi} \hat{Q^2} \ket{\xi} - (\bra{\xi} \hat{Q} \ket{\xi})^2, \\
\Var [\hat{P}] &:= \bra{\xi} \hat{P^2} \ket{\xi} - (\bra{\xi} \hat{P} \ket{\xi})^2.
\end{align}
Here, the probe Hamiltonian is assumed as 
\begin{equation}
	\hat{H} = \frac{\hat{P}^2}{2 m} + V (Q), \label{pham} 
\end{equation}
where $V(Q)$ is the potential on the coordinate space. 
\end{thm}
\begin{proof}
For the probe observable $\hat{M}$, we obtain 
\begin{align}
	\frac{\bra{\xi^\prime} \hat{M} \ket{\xi^\prime}}{\bk{\xi^{\prime}}} & = \frac{\bra{\xi} \hat{M} \ket{\xi} 
	- i g \wv{A}{w} \bra{\xi} \hat{M} \hat{P} \ket{\xi} + i g \overline{\wv{A}{w}} 
	\bra{\xi} \hat{P} \hat{M} \ket{\xi}}{\bk{\xi} 
	- i g \wv{A}{w} \bra{\xi} \hat{P} \ket{\xi} + i g \overline{\wv{A}{w}} 
	\bra{\xi} \hat{P} \ket{\xi}} \notag \\
	& = \frac{\bra{\xi} \hat{M} \ket{\xi} + i g \Re \wv{A}{w} \bra{\xi}[ \hat{P}, 
	\hat{M}]\ket{\xi} + g \Im \wv{A}{w} \bra{\xi} \{ \hat{P}, \hat{M} \} 
	\ket{\xi}}{\bk{\xi} + 2 g \Im \wv{A}{w} \bra{\xi} \hat{P} \ket{\xi}} \notag \\
	& = \left( \bra{\xi} \hat{M} \ket{\xi} + i g \Re \wv{A}{w} \bra{\xi}[ \hat{P}, 
	\hat{M}]\ket{\xi} + g \Im \wv{A}{w} \bra{\xi} \{ \hat{P}, \hat{M} \} 
	\ket{\xi} \right) \notag \\ & \ \ \ \ \ \ \ \ \ \ \ \ \ \ \ \ \ \ \ \ \ \ \ \ \ \ \ \ \ \ \ \ \
	 \ \ \ \ \ \ \ \ \ \ \ 
	\times \left( 1 - 2 g \Im \wv{A}{w} \bra{\xi} \hat{P} \ket{\xi} \right) + O(g^2) \notag \\
	& = \bra{\xi} \hat{M} \ket{\xi} + i g \Re \wv{A}{w} \bra{\xi}[ \hat{P}, 
	\hat{M}]\ket{\xi} \notag \\ & \ \ \ \ \ \ \ \ + g \Im \wv{A}{w} \left( \bra{\xi} \{ \hat{P}, \hat{M} \} 
	\ket{\xi} - 2 \bra{\xi} \hat{M} \ket{\xi} \bra{\xi} \hat{P} \ket{\xi} \right) + O(g^2).
\end{align}
If we set $\hat{M} = \hat{P}$, one has 
\begin{equation}
	\Delta [\hat{P}] = 2 g \Im \wv{A}{w} \Var [\hat{P}].
\end{equation}
If instead we set $\hat{M} = \hat{Q}$, one has 
\begin{equation}
	\Delta [\hat{Q}] = g \Re \wv{A}{w} + g \Im \wv{A}{w} \left( \bra{\xi} \{ \hat{P}, \hat{Q} \} 
	\ket{\xi} - 2 g \bra{\xi} \hat{Q} \ket{\xi} \bra{\xi} \hat{P} \ket{\xi} \right) \label{q21}
\end{equation}
since $[\hat{P},\hat{Q}] = -i$. From the Heisenberg equation with the probe Hamiltonian (\ref{pham}), 
we obtain the Ehrenfest theorem;
\begin{align}
	i \frac{d}{dt} \bra{\xi} \hat{Q} \ket{\xi} & = \bra{\xi} [ \hat{Q}, \hat{H}] \ket{\xi} 
	= i \frac{\bra{\xi} \hat{P} \ket{\xi}}{m} \\
	i \frac{d}{dt} \bra{\xi} \hat{Q}^2 \ket{\xi} & = \bra{\xi} [ \hat{Q}^2, \hat{H}] \ket{\xi}
	= i \frac{\bra{\xi} \{ \hat{P}, \hat{Q} \} \ket{\xi}}{m}.
\end{align}
Substituting them into Eq.~(\ref{q21}), we derive
\begin{equation}
	\Delta [\hat{Q}] = g \Re \wv{A}{w} + m g \Im \wv{A}{w}  
	\left. \frac{d \Var [\hat{Q}]}{dt} \right|_{t=t_0}
\end{equation}
since the interaction to the target system is taken at time $t=t_0$.
\end{proof}
Putting together, we can measure the weak value $\wv{A}{w}$ by observing the 
shift of the expectation value of the probe both in the coordinate and momentum representations. 
The shift of the probe position contains the future information up to the post-selected state.
\begin{col}
	When the probe wavefunction is real-valued in the coordinate representation, Eq.~(\ref{qshift}) can be reduced to 
	\begin{equation}
		\Delta [\hat{Q}] = g \Re \wv{A}{w}.
	\end{equation}
\end{col}
\begin{proof}
	From the Schr\"{o}dinger equation in the coordinate representation; 
	\begin{equation}
		i \frac{\partial}{\partial t} \xi (Q) = \frac{1}{2 m} \frac{\partial^2}{\partial Q^2} \xi (Q) + V (Q) \xi (Q),
	\end{equation}
	where $\xi(Q) \equiv \bkt{Q}{\xi}$, putting $\xi (Q) = R(Q) e^{i S(Q)}$, we obtain the equation for the real part as 
	\begin{equation}
		\frac{\partial}{\partial t} R(Q) + \frac{\partial}{\partial Q} \left( \frac{R(Q) \frac{\partial}{\partial Q} 
		S(Q)}{m} \right) = 0.
	\end{equation}
	Therefore, if the probe wavefunction is real-valued in the coordinate representation, 
	one has $\frac{\partial}{\partial Q} S(Q) = 0$ to obtain $\frac{\partial}{\partial t} R = 0$. Therefore, we obtain 
	\begin{equation}
		\frac{d \Var [\hat{Q}]}{dt} = 0
	\end{equation}
	for any time $t$. Vice versa. From this statement, we obtain the desired result from Eq.~(\ref{qshift}).
\end{proof}
It is noted that there are many analyses on the weak measurement, e.g., on the phase 
space~\cite{Lobo09}, on the finite sample~\cite{Tollaksen07_2}, on the counting statics~\cite{Berry10, Lorenzo11},  
on the non-local observable~\cite{Brodutch08, Brodutch09}, and on the complementary observable~\cite{Wu11}.

Summing up this section, the two-state vector formalism is called if the pre- and post-selected states are prepared and 
the weak or strong measurement is taken in the von-Neumann type Hamiltonian, $H = g A \hat{P} \delta (t - t_0)$ 
between the pre- and post-selected states. In the case of the strong measurement, 
we obtain the expectation value $\Ex (A)$ in the probe. On the other hand, in the case of the weak measurement, 
we obtain the weak value $\wv{A}{w}$ in the probe.
\section{Weak-Value Measurement for a Qubit System}
In this subsection, we consider the weak measurement 
in the case that the probe system is a qubit system~\cite{Wu}. 
In general, the interaction Hamiltonian is given by 
\begin{equation}
	H_{int} = g [ A \otimes (\vec{v} \cdot \vec{\sigma})] \df (t-t_0),
\end{equation}
where $\vec{v}$ is a unit vector. Expanding the interaction Hamiltonian for the pre- and post-selected 
states, $\ket{\psi}$ and $\ket{\phi}$, respectively up to the first order for $g$, we obtain the shift of 
the expectation value for $\vec{q} \cdot \vec{\sigma}$ as 
\begin{align}
	\Delta [\vec{q} \cdot \vec{\sigma}] & = \frac{\bra{\xi^{\prime}} [\vec{q} \cdot \vec{\sigma}] 
	\ket{\xi^{\prime}}}{\bk{\xi^{\prime}}} - \bra{\xi} [\vec{q} \cdot \vec{\sigma}] 
	\ket{\xi} \notag \\
	& = g \bra{\xi} i [\vec{v} \cdot \vec{\sigma}, \vec{q} \cdot \vec{\sigma}] 
	\ket{\xi} \Re \wv{A}{w} \notag \\
	& \ \ \ + g \left( \bra{\xi} \left\{ \vec{v} \cdot \vec{\sigma}, \vec{q} \cdot \vec{\sigma} \right\} 
	\ket{\xi} - 2 \bra{\xi} \vec{v} \cdot \vec{\sigma} \kb{\xi} \vec{q} \cdot \vec{\sigma} \ket{\xi} 
	\right) \Im \wv{A}{w} + O(g^2) \notag \\
	& = 2 g \{ (\vec{q} \times \vec{v}) \cdot \vec{m} \} \Re \wv{A}{w} + 2 g \{ \vec{v} \cdot \vec{q} 
	- (\vec{v} \cdot \vec{m})(\vec{q} \cdot \vec{m}) \} \Im \wv{A}{w} + O(g^2),
\end{align}
where 
\begin{align}
	\ket{\xi^{\prime}} & = \bra{\phi} e^{- i g [ A \otimes (\vec{v} \cdot \vec{\sigma})]} \ket{\psi} 
	\ket{\xi}, \\
	\kb{\xi} & =: \frac{1}{2}(\ID + \vec{m} \cdot \vec{\sigma}).
\end{align}
Furthermore, the pre- and post-selected states are assumed to be 
\begin{equation}
	\kb{\psi} =: \frac{1}{2}(\ID + \vec{r}_i \cdot \vec{\sigma}), \ \ \ 
	\kb{\phi} =: \frac{1}{2}(\ID + \vec{r}_f \cdot \vec{\sigma}).
\end{equation}
Since the weak value of the observable $\vec{n} \cdot \vec{\sigma}$ is 
\begin{equation}
	\wv{\vec{n} \cdot \vec{\sigma}}{w} = \frac{\bra{\phi} \vec{n} \cdot \vec{\sigma} 
	\ket{\psi} \bkt{\psi}{\phi}}{|\bkt{\phi}{\psi}|^2} 
	= \vec{n} \cdot \frac{\vec{r}_i + \vec{r}_f + i (\vec{r}_i \times \vec{r}_f)}{1 + \vec{r}_i \cdot 
	\vec{r}_f},
\end{equation}
we obtain 
\begin{equation}
	\Delta [\vec{q} \cdot \vec{\sigma}] = 2 g \{ (\vec{q} \times \vec{v}) \cdot \vec{m} \} 
	\frac{\vec{n} \cdot (\vec{r}_i + \vec{r}_f)}{1 + \vec{r}_i \cdot 
	\vec{r}_f} + 2 g \{ \vec{v} \cdot \vec{q} 
	- (\vec{v} \cdot \vec{m})(\vec{q} \cdot \vec{m}) \} \frac{\vec{n} \cdot (\vec{r}_i 
	\times \vec{r}_f)}{1 + \vec{r}_i \cdot \vec{r}_f}+ O(g^2). \label{qubitwv}
\end{equation}
From Eq.~(\ref{qubitwv}), we can evaluate the real and imaginary parts of the weak value changing the 
parameter of the measurement direction $\vec{q}$. This calculation is used 
in the context of the Hamiltonian estimation~\cite{ST}.

Next, as mentioned before, we emphasize that the weak measurement is only one of the methods to 
obtain the weak value. There are many other approaches to obtain the weak value, e.g., 
on changing the probe state~\cite{Geszti10, Lorenzo08, Johansen04_2, Nakamura10}, and 
on the entangled probe state~\cite{Menzies08}. Here, we show another method 
to obtain the weak value in the case 
that the target and the probe systems are both qubit systems~\cite{Pryde}. 

Let $\ket{\psi}_s := \alpha \ket{0}_s + \beta \ket{1}_s$ be the pre-selected state for the target system. The 
initial probe state can described as $\ket{\xi}_p := \gamma \ket{0}_p + \eta \ket{1}_p$. It is emphasized that 
the initial probe state is controllable. Here, the initial states 
are normalized, that is, $|\alpha|^2 + |\beta|^2 = 1$ and $|\gamma|^2 + |\eta|^2 = 1$. Applying 
the Controlled-NOT (C-NOT) gate, we make a transform of the quantum state for the compound system to  
\begin{equation}
	\ket{\psi}_s \otimes \ket{\xi}_p \xlongrightarrow[]{{\rm C-NOT}} 
	\ket{\Psi_c} := (\alpha \gamma \ket{0}_s + \beta \eta \ket{1}_s) \ket{0}_p + 
	(\alpha \eta \ket{0}_s + \beta \gamma \ket{1}_s) \ket{1}_p.
\end{equation}
In the case of $\gamma \sim 1$, we obtain the compound state as 
\begin{equation}
	\alpha \ket{0}_s \ket{0}_p + \beta \ket{1}_s \ket{1}_p,
\end{equation}
and similarly, in the case of $\eta \sim 1$, one has  
\begin{equation}
	\alpha \ket{0}_s \ket{1}_p + \beta \ket{1}_s \ket{0}_p.
\end{equation}
Those cases can be taken as the standard von Neumann projective measurement. 
For the post-selected state $\ket{\phi}$, the probability to obtain the measurement outcome $k$ on 
the probe is 
\begin{align}
	\Pr [k] & := \frac{\parallel \left( \,_s \bra{\phi} \otimes \,_p \bra{k} \right) \ket{\Psi_c} 
	\parallel^2}{\sum_{m \in \{ 0, 1 \} } 
	\parallel \left( \,_s \bra{\phi} \otimes \,_p \bra{m} \right) \ket{\Psi_c} 
	\parallel^2} \notag \\
	& = \frac{\left| \left( \,_s \bkt{\phi}{0}_s \bkt{0}{\psi}_s \gamma
	+ \,_s \bkt{\phi}{1}_s \bkt{1}{\psi}_s \eta \right) \delta_{k,0} + 
	\left( \,_s \bkt{\phi}{0}_s \bkt{0}{\psi}_s \eta 
	+ \,_s \bkt{\phi}{1}_s \bkt{1}{\psi}_s \gamma \right) \delta_{k,1} 
	\right|^2}{\sum_{m \in \{ 0, 1 \} } 
	\parallel \left( \,_s \bra
	{\phi} \otimes \,_p \bra{m} \right) \ket{\Psi_c} 
	\parallel^2} \notag \\
	& = \frac{|(\gamma - \eta) \,_s\bkt{\phi}{k}_s\bkt{k}{\psi}_s 
	+ \eta \,_s\bkt{\phi}{\psi}_s|^2}{|(\gamma - \eta) \,_s\bkt{\phi}{0}_s\bkt{0}{\psi}_s 
	+ \eta \,_s\bkt{\phi}{\psi}_s|^2 
	+ |(\gamma - \eta) \,_s\bkt{\phi}{1}_s\bkt{1}{\psi}_s + \eta \,_s\bkt{\phi}{\psi}_s|^2} \notag \\
	& = \frac{|(\gamma - \eta) \,_{\phi}\wv{\ket{k}_s\bra{k}}{\psi}^{w} 
	+ \eta|^2}{1 - (\gamma - \eta)^2 (1 - \sum_{m \in \{ 0,1 \}} 
	|\,_{\phi}\wv{\ket{m}_s\bra{m}}{\psi}^{w}|^2)}.
	\label{propara}
\end{align}
Here, in the last line, the parameters $\gamma$ and $\eta$ are assumed to be real. 
Without the post-selection, the POVM to obtain the measurement outcome $k$ is 
\begin{equation}
	E_k = (\gamma^2 - \eta^2) \ket{k}_s\bra{k} + \eta^2.
\end{equation}
Here, the coefficient of the first term means that the strength of measurement and 
the second term is always added. Therefore, we define the quantity to distinguish 
the probability for the measurement outcome $k$ as 
\begin{equation}
	R [k] := \frac{\Pr [k] - \eta^2}{(\gamma^2 - \eta^2)}. \label{readoutt}
\end{equation}
Putting together Eqs.~(\ref{propara}) and (\ref{readoutt}), we obtain 
\begin{equation}
	R [k] = \frac{2 \eta (\gamma - \eta) \Re \,_{\phi}\wv{\ket{k}_s\bra{k}}{\psi}^{w} + 
	(\gamma - \eta)^2 [|\,_{\phi}\wv{\ket{k}_s\bra{k}}{\psi}^{w}|^2 + \eta^2 (1 - 
	|\,_{\phi}\wv{\ket{k}_s\bra{k}}{\psi}^{w}|^2)]}{(\gamma^2 - \eta^2)[1 - (\gamma - \eta)^2 
	(1 - \sum_{m \in \{ 0,1 \}} |\,_{\phi}\wv{\ket{m}_s\bra{m}}{\psi}^{w}|^2)]}.
\end{equation}
Setting the parameters; 
\begin{equation}
	\gamma = \sqrt{\frac{1}{2} + \epsilon}, \ \ \ 
	\eta = \sqrt{\frac{1}{2} - \epsilon},
\end{equation}
one has 
\begin{align}
	R [k] & = \frac{ (1 - \epsilon) \Re \,_{\phi}\wv{\ket{k}_s\bra{k}}{\psi}^{w} + 
	\epsilon \left[ |\,_{\phi}\wv{\ket{k}_s\bra{k}}{\psi}^{w}|^2 + \left( \frac{1}{2} - \epsilon \right) 
	(1 - |\,_{\phi}\wv{\ket{k}_s\bra{k}}{\psi}^{w}|^2) \right]}{2 \left[ 1 - \epsilon^2  
	\left( 1 - \sum_{m \in \{ 0,1 \}} |\,_{\phi}\wv{\ket{m}_s\bra{m}}{\psi}^{w}|^2 \right) \right]} + O(\epsilon^2), \notag \\
	& = \frac{1}{2} \Re \,_{\phi}\wv{\ket{k}_s\bra{k}}{\psi}^{w} - \frac{\epsilon}{2} 
	\left( \Re \,_{\phi}\wv{\ket{k}_s\bra{k}}{\psi}^{w} - \frac{1}{2} 
	|\,_{\phi}\wv{\ket{k}_s\bra{k}}{\psi}^{w}|^2 \right) + O(\epsilon^2).
	\label{prydemethod}
\end{align}
From Eq.~(\ref{prydemethod}), it is possible to obtain the real part of the weak value from the first term and 
its imaginary part from the second term. Since the first order of the parameter $\epsilon$ is the 
gradient on changing the initial probe state from $\ket{\xi}_p = \frac{1}{\sqrt{2}}
( \ket{0}_p + \ket{1}_p )$, realistically, we can evaluate the imaginary part of the weak value 
from the gradient of the readout. This method is also used in Ref.~\cite{YYKI} on the joint weak value. 
It is emphasized that the weak value can be experimentally accessible by changing the initial probe 
state while the interaction is {\bf not weak}~\footnote{This point seems to be misunderstood. 
According to Ref.~\cite{Pryde2}, the violation of the Leggett-Garg inequality~\cite{LG} was shown, 
but the macroscopic realism cannot be denied since the noninvasive measurability is not realized.}.
\section{Weak Values for Arbitrary Coupling Quantum Measurement}
We just calculate an arbitrary coupling between the target and the probe systems~\cite{Wu_full,Nakamura_full,Koike_full}. 
Throughout this section, we assume that the desired observable is the projection operator to be denoted as $A^2 = A$~\cite{Shikano_full}. 
In the case of the von-Neumann interaction motivated by the original work~\cite{AAV}, when the pre- and post-selected 
states are $\ket{i}$ and $\ket{f}$, respectively, and the probe state is $\ket{\xi}$, the probe state $\ket{\xi^{\prime}}$ 
after the interaction given by $H_{int} = g A \hat{P}$ becomes
\begin{align}
	\ket{\xi^{\prime}} & = \bra{f} e^{-i g A \hat{P}} \ket{i} \ket{\xi} 
	= \bra{f} \left( 1 + \sum_{k=1}^{\infty} \frac{1}{k !}(- i g A \hat{P})^k \right) \ket{i} \ket{\xi} 
	= \bra{f} \left( 1 + A \sum_{k=1}^{\infty} \frac{1}{k !}(- i g \hat{P})^k \right) \ket{i} \ket{\xi} \notag \\
	& = \bra{f} \left( 1 - A + A \sum_{k=0}^{\infty} \frac{1}{k !} (- i g \hat{P})^k \right) \ket{i} \ket{\xi}
	= \bra{f} \left( 1 - A + A e^{-i g \hat{P}} \right) \ket{i} \ket{\xi} \notag \\
	& = \bkt{f}{i} \left( 1 - \wv{A}{w} + \wv{A}{w} e^{-i g \hat{P}} \right) \ket{\xi}. \label{full_order_eq}
\end{align}
It is remarked that the desired observable $B$, which satisfies $B^2 = 1$~\cite{Nakamura_full, Koike_full}, corresponds to $B = 2 A - 1$. 
Analogous to Theorem~\ref{Jozsaproof}, we can derive the expectation values of the position and the momentum after the weak measurement. 
These quantities depends on the weak value $\wv{A}{w}$ and the generating function for the position and the momentum of the initial probe state $\ket{\xi}$.
\section{Weak Value with Decoherence} \label{wvdec_sec}
The decoherence results from the coupled system to the environment and leads to 
the transition from the quantum to classical systems. The general framework of the 
decoherence was discussed in Sec.~\ref{qo_sec}. In this section, we discuss the analytical expressions 
for the weak value. 

While we directly discuss the weak value with decoherence, the weak value is defined 
as a complex number. To analogously discuss the density operator formalism, we need the 
operator associated with the weak value. Therefore, we define a {\it W operator} $W(t)$ as
\begin{equation} 
	W(t) := U(t,t_i)\ket{i}{\bra{f}}U(t_f,t).
	\label{defwo}
\end{equation} 
To facilitate the formal development of the weak value, we introduce the ket state 
$\ket{\psi(t)}$ and the bra state $\bra{\phi(t)}$ as
\begin{equation}
	\ket{\psi(t)}= U(t,t_i)\ket{i} , \ \bra{\phi(t)} =\bra{f} U(t_f,t),
\end{equation}
so that the expression for the W operator simplifies to
\begin{equation}
	W(t) = \kbt{\psi(t)}{\phi(t)}. \label{tanweak}
\end{equation} 
By construction, the two states $\ket{\psi(t)}$ and $\bra{\phi(t)}$ satisfy the 
Schr\"{o}dinger equations with the same Hamiltonian with the initial and final 
conditions $\ket{\psi(t_i)}=\ket{i}$ and $\bra{\phi(t_f)}=\bra{f}$. In a sense, 
$\ket{\psi(t)}$ evolves forward in time while $\bra{\phi(t)}$ evolves backward in time. 
The time reverse of the W operator (\ref{tanweak}) is $W^{\dagger} = \kbt{\phi (t)}{\psi (t)}$. 
Thus, we can say the W operator is based on the two-state vector formalism 
formally described in Refs.~\cite{AV90, AV91}. 
Even an apparently similar quantity to the W operator (\ref{tanweak}) was introduced by 
Reznik and Aharonov~\cite{Reznik95} in the name of ``two-state" with the conceptually 
different meaning. This is because the W operator acts on a Hilbert space 
$\HH$ but the two-state vector acts on the Hilbert space 
$\overrightarrow{\HH_1} \otimes \overleftarrow{\HH_2}$. Furthermore, while the generalized 
two-state, which is called a multiple-time state, was introduced~\cite{APTV}, 
this is essentially reduced to the two-state vector formalism.
The W operator gives the weak value of the observable $A$~\footnote{While the original notation 
of the weak values is $\wv{A}{w}$ indicating the ``w"eak value of an observable $A$, 
our notation is motivated by one of which the pre- and post-selected states are 
explicitly shown as $\,_{f} \wv{A}{i}^{w}$.} as
\begin{equation}
	\wv{A}{W} = \frac{\Tr (WA)}{\Tr W}, 
\end{equation}
in parallel with the expectation value of the observable $A$ by 
\begin{equation}	
	\Ex [A] = \frac{\Tr (\rho A)}{\Tr \rho} 
\end{equation}
from Born's rule. Furthermore, the W operator (\ref{defwo}) can be regarded as 
a special case of a standard purification of the density operator~\cite{UHLMANN86}. 
In our opinion, the W operator should be considered on the same footing of the 
density operator. For a closed system, both satisfy the Schr{\"o}dinger equation. 
In a sense, the W operator $W$ is the square root of the density operator since 
\begin{equation}
	W(t)W^{\dagger}(t) = \ket{\psi(t)}\bra{\psi(t)} = U(t,t_i) \kb{i} U^{\dagger}(t,t_i),
\end{equation} 
which describes a state evolving forward in time for a given initial state $\ket{\psi(t_i)}\bra{\psi(t_i)}=\ket{i}\bra{i}$, while
\begin{eqnarray}
	W^{\dagger}(t)W(t) = \ket{\phi(t)}\bra{\phi(t)} = U(t_f,t) \kb{f} U^{\dagger}(t_f,t),
\end{eqnarray} 
which describes a state evolving backward in time for a given final state 
$\ket{\phi(t_f)}\bra{\phi(t_f)}=\ket{f}\bra{f}$. The W operator describes the entire history of 
the state from the past ($t_i$) to the future ($t_f$) and measurement 
performed at the time $t_0$ as we shall see in Appendix~\ref{WMR}. 
This description is conceptually different from the conventional one by the 
time evolution of the density operator. From the viewpoint of geometry, the W operator can be taken 
as the Hilbert-Schmidt bundle. The bundle projection is given by 
\begin{equation} 
	\Pi: W(t) \to \rho_i (t) := W(t) W^{\dagger}(t).
\end{equation}
When the dimension of the Hilbert space is $N$: ${\rm dim} \HH = N$, the structure 
group of this bundle is $U(N)$~\cite[Sec. 9.3]{Bengtsson}. 
Therefore, the W operator has richer information than the density 
operator formalism as we shall see a typical example of a geometric phase~\cite{SH}.
Furthermore, we can express the probability to get the measurement outcome $a_n \in A$ due to 
the ABL formula (\ref{ablrule}) using the W operator $W$ as 
\begin{equation}
	\Pr [A = a_n] = \frac{|\Tr W P_{a_n}|^2}{\sum_n |\Tr W P_{a_n}|^2},
\end{equation}
where $A = \sum_n a_n \kb{a_n} =: \sum_n a_n P_{a_n}$. This shows the usefulness of the W operator.

Let us discuss a state change in terms of the W operator and define a map ${\cal X}$ as 
\begin{equation}
	{\cal X} (\ket{\alpha}, \ket{\beta}) := (\EE \otimes \ID) \left( \kbt{\alpha}{\beta} \right),
	\label{weakop}\\
\end{equation}
for an arbitrary $\ket{\alpha}, \ket{\beta} \in \HH_s \otimes \HH_e$. Then, we obtain the following theorem on the change of the W operator such as Theorem \ref{1th}.
\begin{thm}
For any W operator $W = \ket{\psi (t)}_{s} \bra{\phi (t)}$, we expand 
\begin{equation}
	\ket{\psi (t)}_{s} = \sum_{m} \psi_{m} \ket{\alpha_m}_{s}, \ \ket{\phi (t)}_{s} = 
	\sum_{m} \phi_{m} \ket{\beta_m}_{s}, 
\end{equation}
with fixed complete orthonormal sets $\{\ket{\alpha_m}_{s} \}$ and $\{\ket{\beta_m}_{s} \}$.
Then, a change of the W operator can be written as  
\begin{equation}
	\EE \left( \ket{\psi (t)}_{s} \bra{\phi (t)} \right)= \, _{e}\bra{\tilde{\psi} (t)} {\cal X} 
	(\ket{\alpha}, \ket{\beta}) \ket{\tilde{\phi} (t)}_{e},
	\label{tvwo}
\end{equation}
where
\begin{equation}
	\ket{\tilde{\psi} (t)}_{e} = \sum_{k} \psi^{\ast}_{k} \ket{\alpha_k}_{e}, \ 
	\ket{\tilde{\phi} (t)}_{e} = \sum_{k} \phi^{\ast}_{k} \ket{\beta_k}_{e}, \label{tildetwo}
\end{equation}
and $\ket{\alpha}$ and $\ket{\beta}$ are maximally entangled states defined by 
\begin{equation}
	\ket{\alpha} := \sum_m \ket{\alpha_m}_s \ket{\alpha_m}_e, \ 
	\ket{\beta} := \sum_m \ket{\beta_m}_s \ket{\beta_m}_e.
\end{equation}
Here, $\{\ket{\alpha_m}_{e} \}$ and $\{\ket{\beta_m}_{e} \}$ are complete orthonormal sets corresponding to $\{\ket{\alpha_m}_{s} \}$ and $\{\ket{\beta_m}_{s} \}$, respectively.
\label{e2th}
\end{thm}
The proof is completely parallel to that of Theorem \ref{1th}.
\begin{thm}
	For any W operator $W = \ket{\psi (t)}_{s} \bra{\phi (t)}$, given the CP map $\EE$, the operator-sum 
	representation is written as 
	\begin{equation}
		\EE(W)=\sum_{m} E_{m}WF^{\dagger}_{m}, 
		\label{wk}
	\end{equation}
	where $E_m$ and $F_m$ are the Kraus operators.
\end{thm}
It is noted that, in general, 
$\EE (W) \EE(W^{\dagger}) \neq \EE(\rho)$ although $\rho = WW^{\dagger}$.
\begin{proof}
We take the polar decomposition of the map $X$ to obtain
\begin{equation}
	{\cal X} = \KK u,
\end{equation}
noting that 
\begin{equation}
	{\cal X} {\cal X}^{\dagger} = \KK u u^{\dagger} \KK = \KK^2.
\end{equation}
The unitary operator $u$ is well-defined on $\HH_s \otimes \HH_e$ because $\KK$ 
defined in Eq. (\ref{sig}) is positive. This is a crucial point to obtain this result (\ref{wk}), 
which is the operator-sum representation for the quantum operation of the W operator. 
From Eq. (\ref{complete}), we can rewrite ${\cal X}$ as 
\begin{equation}
	{\cal X} = \sum_{m} \kb{s_m} u 
	= \sum_{m} \kbt{s_m}{t_m},
\end{equation}
where
\begin{equation}
	\bra{t_{m}} = \bra{s_m} u.
\end{equation}
Similarly to the Kraus operator (\ref{krausreview}), we define the two operators, $E_m$ 
and $F^{\dagger}_{m}$, as 
\begin{equation}
	E_{m} \ket{\psi (t)}_s := \,_{e}\bkt{\tilde{\psi} (t)}{s_m}, \ \ \ \,_{s}\bra{\phi (t)} F^{\dagger}_{m} := \bkt{t_{m}}{\tilde{\phi} (t)}_{e},
\end{equation}
where $\ket{\tilde{\psi} (t)}_e$ and $\ket{\tilde{\phi} (t)}_e$ are defined in Eq. (\ref{tildetwo}).
Therefore, we obtain the change of the W operator as 
\begin{align}
	\sum_{m} E_{m} \ket{\psi (t)}_s \bra{\phi (t)} F^{\dagger}_{m} 
	& = \sum_{m} \,_{e}\bkt{\tilde{\psi}(t)}{s_m} \bkt{t_{m}}{\tilde{\phi}(t)}_{e} 
	= \,_{e}\bra{\tilde{\psi}(t)} {\cal X} \ket{\tilde{\phi}(t)}_{e} \notag \\
	& = \EE \left( \ket{\psi (t)}_s \bra{\phi (t)} \right),
\end{align}
using Theorem \ref{e2th} in the last line. By linearity, we obtain the desired result.
\end{proof}

Summing up, we have introduced the W operator (\ref{defwo}) and obtained the general 
form of the quantum operation of the W operator (\ref{wk}) in an analogous way to 
the quantum operation of the density operator assuming the complete positivity 
of the physical operation. This can be also described from information-theoretical 
approach~\cite{Chiribella09} to solve the open problem listed in Ref.~\cite[Sec. XII]{APTV}. 
However, this geometrical meaning has still been an open problem.

It is well established that the trace preservation, $\Tr(\EE(\rho))= \Tr \rho = 1$ for all $\rho$, 
implies that $\sum_{m}E^{\dagger}_{m}E_{m}=1$. As discussed in Eq. (\ref{preconsist}), the proof goes through as  
\begin{equation}
1 =\Tr(\EE(\rho)) = \Tr \left( \sum_{m}E_{m}\rho E^{\dagger}_{m} \right) 
= \Tr \left( \sum_{m}E^{\dagger}_{m}E_{m}\rho \right) \; ( \forall \rho ).
\end{equation}
This argument for the density operator $\rho= WW^{\dagger}$ applies also for $W^{\dagger} W$ 
to obtain $\sum_{m}F^{\dagger}_{m}F_{m}=1$ because this is the density operator in the time 
reversed world in the two-state vector formulation as reviewed in Sec.~\ref{wvhis_sec}. 
Therefore, we can express the Kraus operators, 
\begin{equation}
	E_{m} = \,_e\bra{e_m}U\ket{e_i}_e, \ F_{m}^{\dagger} = \,_e\bra{e_f}V\ket{e_m}_e, \label{kraus2}
\end{equation}
where 
\begin{equation}
	U = U (t, t_i), \ V = U(t_f,t), \label{twounitary}
\end{equation}
are the evolution operators, which act on $\HH_s \otimes \HH_e$. 
$\ket{e_i}$ and $\ket{e_f}$ are some basis vectors and  $\ket{e_m}$ is 
a complete set of basis vectors with $\sum_{m} \kb{e_m} = 1$. We can compute 
\begin{equation}
\sum_{m}F^{\dagger}_{m}E_{m} = \sum_{m} \,_e\bra{e_f}V\ket{e_m}_e \bra{e_m}U\ket{e_i}_e 
= \,_e\bra{e_f}VU\ket{e_i}_e.
\label{Sma}
\end{equation}
The above equality (\ref{Sma}) may be interpreted as a decomposition of the history in analogy to the decomposition of unity because
\begin{equation}
	\,_e\bra{e_f}VU\ket{e_i}_e = \, _e\bra{e_f}S\ket{e_i}_e = S_{fi} 
	\label{smatrix}
\end{equation}
is the S-matrix element. On this idea, Ojima and Englert have developed the formulation on the S-matrix 
in the context of the algebraic quantum field theory~\cite{Ojima} and the backaction of the Hawking 
radiation~\cite{Englert10}, respectively.
\section{Weak Measurement with Environment} \label{known}
Let us consider a target system coupled with an environment and a general weak measurement 
for the compound of the target system and the environment. We assume that there is no interaction 
between the probe and the environment and the same interaction between the target and probe systems 
(\ref{wmint}). The Hamiltonian for the target system and the environment is given by
\begin{equation}
	H = H_0 \otimes \ID_e + H_1,
\end{equation}
where $H_0$ acts on the target system $\HH_s$ and the identity operator $\ID_e$ is 
for the environment $\HH_e$, while $H_1$ acts on $\HH_s\otimes \HH_e $. 
The evolution operators $U := U(t,t_i)$ and $V := U(t_f,t)$ as defined 
in Eq. (\ref{twounitary}) can be expressed by
\begin{equation}
	U = U_{0} K(t_0,t_i), \ 
	V = K(t_f,t_0) V_{0},
\end{equation}
where $U_{0}$ and $V_{0}$ are the evolution operators forward in time and backward in time, 
respectively, by the target Hamiltonian $H_0$.
$K$'s are the evolution operators in the interaction picture,
\begin{equation}
	 K(t_0,t_i) = \TT e^{-i\int^{t_0}_{t_i} dt U_{0}^{\dagger} H_1 U_{0}}, \ 
	 K(t_f,t_0) = \overline{\TT} e^{-i\int^{t_f}_{t_0} dt V_{0} H_1 V_{0}^{\dagger}},
\end{equation}
where $\TT$ and $\overline{\TT}$ stand for the time-ordering and anti time-ordering products.

Let the initial and final environmental states be $\ket{e_i}$ and $\ket{e_f}$, respectively. 
The probe state now becomes
\begin{equation}
	\ket{\xi^\prime} = \bra{f}\bra{e_f}VU\ket{e_i}\ket{i} \left( \ID - g\frac{\bra{f}\bra{e_f}VAU\ket{e_i}
	\ket{i}}{\bra{f}\bra{e_f}VU\ket{e_i}\ket{i}} \hat{P} + O (g^2) \right) \ket{\xi}.
\end{equation}
Plugging the expressions for $U$ and $V$ into the above, we obtain the probe state as
\begin{equation}
	\ket{\xi^\prime} = N \xi \left( \ID - g \frac{\bra{f}\bra{e_f}K(t_f,t_0)
	V_{0}AU_{0}K(t_0,t_i)\ket{e_i}\ket{i}}{N} \hat{P} \right) \ket{\xi} + O (g^2) ,
\end{equation}
where $N= \bra{f}\bra{e_f}K(t_f,t_0)V_{0}U_{0}K(t_0,t_i)\ket{e_i}\ket{i}$
is the normalization factor. We define the dual quantum operation as
\begin{equation}
	\EE^{\ast}(A) :=\bra{e_f}K(t_f,t_0)V_{0}AU_{0}K(t_0,t_i)\ket{e_i} =\sum_{m} V_{0} F^{\dagger}_{m}A E_m U_{0},
\end{equation}
where 
\begin{align}
	F^{\dagger}_{m} & := V^{\dagger}_{0}\bra{e_f}K(t_f,t_0)\ket{e_m}V_0, \\
	E_m & := U_{0}\bra{e_m}K(t_0,t_i)\ket{e_i}U^{\dagger}_{0}
\end{align}
are the Kraus operators. Here, we have inserted the completeness 
relation $\sum_{m}\ket{e_m}\bra{e_m}=1$ with $\ket{e_m}$ being not necessarily orthogonal. 
The basis $\ket{e_i}$ and $\ket{e_f}$ are the initial and final environmental states, 
respectively.  
Thus, we obtain the wave function of the probe as
\begin{align}
\ket{\xi^\prime} & = N \left( \ID -g\frac{\bra{f}\EE^{*}(A)\ket{i}}{N} \hat{P} \right) \ket{\xi} 
+ O (g^2) \notag \\ 
&= N \left( \ID - g \frac{\sum_{m} \bra{f}V_0 F^{\dagger}_{m} A E_{m} U_0 \ket{i}}{\sum_{m} 
\bra{f} V_0 F^{\dagger}_m E_m U_0 \ket{i}} \hat{P} \right) \ket{\xi} + O (g^2) \notag \\
& = N \left( \ID - g \frac{\Tr \left[ A \sum_m E_m U_0 \kbt{i}{f} 
V_0 F^{\dagger}_m \right]}{\Tr \left[ \sum_m E_m U_0 \kbt{i}{f} V_0 F^{\dagger}_m \right]} \hat{P} 
\right) \ket{\xi} + O(g^2) \notag \\
& = N \left( \ID - g \frac{\Tr[\EE(W)A]}{\Tr [\EE (W)]} \hat{P} \right) \ket{\xi} + O(g^2)= 
N (\ID -g\wv{A}{\EE(W)} \hat{P}) \ket{\xi} + O(g^2),
\end{align}
Analogous to Theorem~\ref{Jozsaproof}, 
the shift of the expectation value of the position operator on the probe is
\begin{equation} \label{eshiftq}
	\Delta[Q] = g \cdot \Re[\wv{A}{\EE(W)}] + m g \cdot \Im[\wv{A}{\EE(W)}] \left. \frac{d 
	\Var [Q]}{dt} \right|_{t=t_0}.
\end{equation} 
From an analogous discussion, we obtain the shift of the expectation value of 
the momentum operator on the probe as 
\begin{equation} \label{eshiftp}
	\Delta[P] = 2 g \cdot \Var [P] \cdot \Im[\wv{A}{\EE(W)}].
\end{equation}
Thus, we have shown that the probe shift in the weak measurement is exactly given 
by the weak value defined by the quantum operation of the W operator due to the environment.
\section{Summary}
We have reviewed that the weak value is defined independent of the weak measurement in the original idea~\cite{AAV} and  
have explained its properties. Furthermore, to extract the weak value, we have constructed some measurement model to extract the weak value. 
I hope that the weak value becomes the fundamental quantity to describe quantum mechanics and quantum field theory and has practical advantage in the quantum-mechanical world. 
\section*{Acknowledgment}
The author acknowledges useful collaborations and discussion with Akio Hosoya, Yuki Susa, and Shu Tanaka. The author thanks Yakir Aharonov, Richard Jozsa, Sandu Popescu, Aephraim Steinberg, and Jeff Tollaksen for useful discussion. The author would like to thank the use of the utilities of Tokyo Institute of Technology and Massachusetts Institute of Technology and many technical and secretary supports. The author is grateful to the financial supports from JSPS Research Fellowships for Young Scientists (No. 21008624), JSPS Excellent Young Researcher Overseas Visit Program, Global Center of Excellence Program ``Nanoscience and Quantum Physics" at Tokyo Institute of Technology during his Ph.D study.

\end{document}